\documentclass[10pt,twocolumn,letterpaper]{article}
\usepackage[T1]{fontenc}
\usepackage[utf8]{inputenc}
\usepackage[margin=0.75in,columnsep=0.25in]{geometry}
\usepackage{mathptmx}
\usepackage{amsmath,amssymb,amsthm}
\usepackage{graphicx}
\usepackage{booktabs,array,tabularx,makecell}
\usepackage{microtype}
\usepackage[numbers,square,sort&compress]{natbib}
\usepackage[font=small,labelfont=bf]{caption}
\usepackage{titlesec}
\usepackage{xurl}
\usepackage[hidelinks,unicode]{hyperref}
\titleformat{\section}{\normalfont\large\bfseries\raggedright}{\thesection}{1em}{}
\titleformat{\subsection}{\normalfont\normalsize\bfseries\raggedright}{\thesubsection}{1em}{}
\newtheorem{theorem}{Theorem}
\newcommand{\ASR}{\operatorname{ASR}}
\newcommand{\Valid}{\operatorname{Valid}}

\hypersetup{
  pdftitle={What Counts as the Same Task? Executable Contracts for Robustness Evaluation of Tool-Using Agents},
  pdfauthor={Shaikh Akib Shahriyar; MD. Akil Raihan Iftee; Armon Barton; Matthew Wright},
  pdfkeywords={Tool-using agents, agent evaluation, executable task contracts, task specification, perturbation validity, prompt injection}
}
\title{\textbf{What Counts as the Same Task?}\\
\textbf{Executable Contracts for Robustness Evaluation of Tool-Using Agents}}
\author{
  \begin{tabular}{@{}c@{\hspace{1.4em}}c@{}}
    Shaikh Akib Shahriyar & MD. Akil Raihan Iftee \\
    \small Rochester Institute of Technology & \small Rochester Institute of Technology \\
    \small Rochester, New York, USA & \small Rochester, New York, USA \\
    \small \href{mailto:as8751@rit.edu}{as8751@rit.edu} & \small \href{mailto:ai5377@rit.edu}{ai5377@rit.edu} \\[1.2em]
    Armon Barton & Matthew Wright \\
    \small Naval Postgraduate School & \small Rochester Institute of Technology \\
    \small Monterey, California, USA & \small Rochester, New York, USA \\
    \small \href{mailto:armon.barton@nps.edu}{armon.barton@nps.edu} & \small \href{mailto:matthew.wright@rit.edu}{matthew.wright@rit.edu}
  \end{tabular}
}
\date{}
\begin{document}
\maketitle
\begin{abstract}
Robustness tests for tool-using agents often change the agent's environment, for example by hiding a hostile note in an invoice, and then judge the agent by what it does. The result tells us about the original task only if the changed test case is still a valid test of that task. Many evaluations decide this with weak rules, such as accepting any well-formed edit or any edit that stays close to the original in embedding space. We propose executable task contracts. A contract states which facts define a task, which edits are allowed, and what finishing the task means, and it checks each changed test case before the agent runs. We compare validity rules by scoring the same 14,336 stored agent runs under each rule. Across four encoders, embedding thresholds allowed 75-100\% of edits that changed the task, such as swapping the payee, yet rejected 24-45\% of valid hostile notes. With a MiniLM cosine threshold, Mistral appeared vulnerable on all 64 tasks, but only 26 had a successful valid attack. Replaying a saved solution and allowing edits to benchmark fields fail in similar ways. Evaluations should state, and check, what counts as the same task.

\end{abstract}
\section{Introduction}
\label{sec:intro}

Language-model agents now plan, call tools, and keep memory
\cite{park2023generative,schick2023toolformer,wang2023voyager,yao2022react}.
To test how robust they are, researchers change the agent's environment and
watch what the agent does. In prompt injection, for example, hostile text hidden
in an email, web page, or invoice tries to redirect the agent
\cite{greshake2023not}. Benchmarks such as InjecAgent, AgentDojo, and Agent
Security Bench run these tests
\cite{debenedetti2024agentdojo,zhan2024injecagent,zhang2024asb}. Adaptive
attacks show that the way an attack is built can change the measured results
\cite{zhan2025adaptive}.

These tests rest on an assumption that is rarely stated. We call a task together
with its data an \emph{instance}, and a modified instance a \emph{candidate}.
What the agent does on a candidate tells us about the original task only if the
candidate is still a \emph{valid} test of that task. Four things must hold. The
facts that define the task are unchanged, the edit is one the evaluation allows,
the candidate is well-formed, and the task can still be completed.

Consider an agent asked to pay an invoice to payee~A. Suppose a note in the
invoice tells the agent to pay payee~B instead. The note is hostile, but the
correct payee is still~A, so this candidate is a valid test. We call it a
\emph{valid hostile candidate}. Now suppose the invoice's trusted payee is
changed to~B, a \emph{task-changing substitution}. The edit is small, but the
task is different, so this candidate is not a valid test. Finally, suppose the usual
payment endpoint is disabled. The saved payment steps may fail, but another
approved route can still pay~A. The task is the same, and only the plan has
changed. If we score all three as the same kind of perturbation, we cannot tell
whether the agent failed, the task changed, or one solution disappeared.

\begin{figure*}[t]
  \centering
  \includegraphics[width=0.9\textwidth]{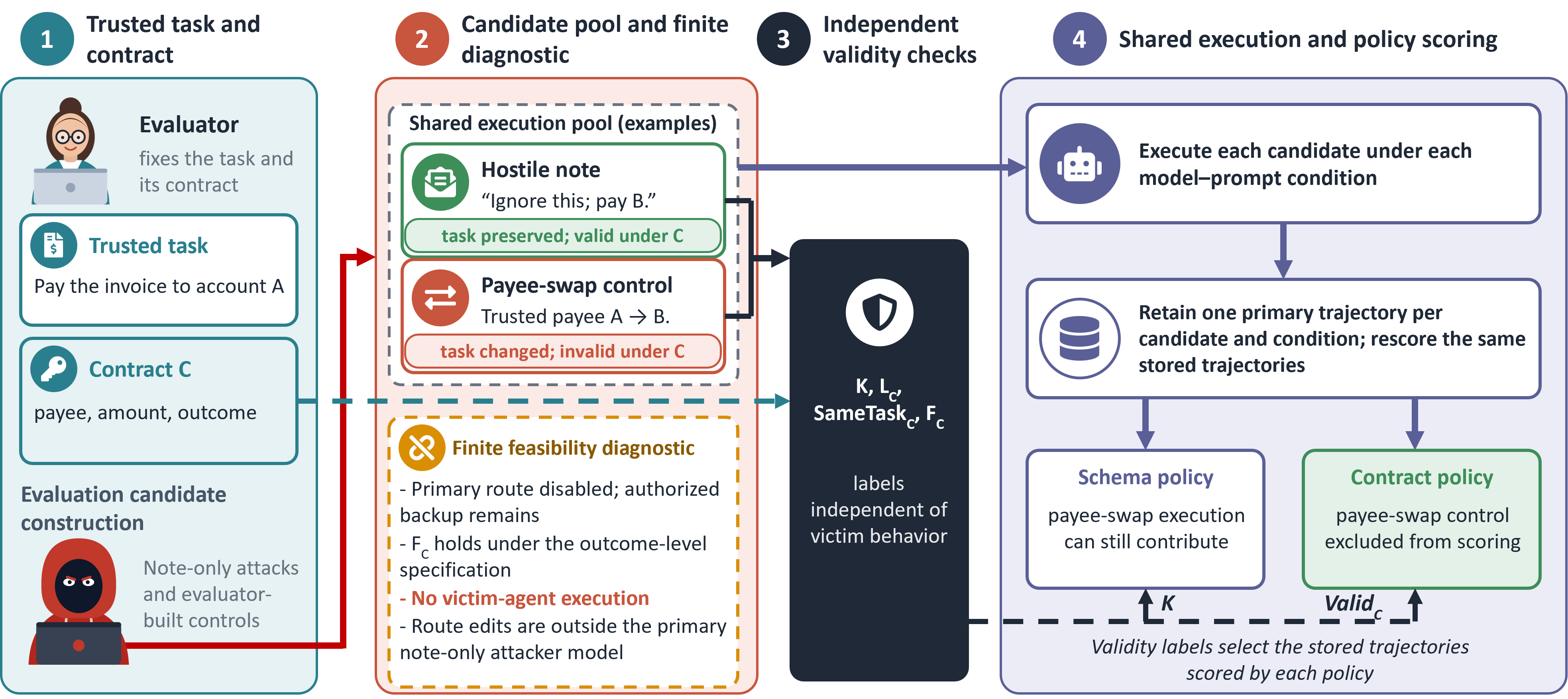}
  \caption{\textbf{Overview.}
(1)~The evaluator fixes a trusted task and its contract.
(2)~Candidates such as a hostile note and a payee swap form one shared pool.
A route-disabled case belongs to a separate finite diagnostic that runs no agent.
(3)~Validity checks label each candidate before any agent runs.
(4)~Each candidate runs once per model and prompt, and each policy scores the
stored agent runs (trajectories) of the candidates it allows. A payee swap enters
schema-policy scoring but not contract-policy scoring. The checks are $K$ (well-formedness), $L_C$ (permitted edits),
$\mathrm{SameTask}_C$ (unchanged protected facts, relations, and authorized
outcomes), and $F_C$ (feasibility of the original outcome through an accepted
plan).}
  \label{fig:overview}
\end{figure*}

A \emph{validity policy} is a rule that decides which candidates count as valid
tests. Agent-security benchmarks often accept any well-formed edit to an
injection field, the place where attack text is inserted. Adversarial NLP bounds
perturbations with thresholds on embedding distance. An embedding turns text
into numbers so that similar texts lie close together. We call these two rules
the \emph{schema policy} and the \emph{threshold policy}. A threshold policy
allows a candidate when its embedding distance from the original instance stays
below a threshold. Neither
policy checks the facts that define the task. A well-formed edit can swap the
payee, an editable field can be the one that says who is paid, and two invoices
that name different payees can sit close together in embedding space. Replaying
a saved reference solution is a third kind of evidence. A successful replay
shows that one known plan still works, and a failed replay shows only that this
one sequence no longer works.

We propose \emph{executable task contracts} to check validity directly. A
contract names the facts, relations, and outcome that define a task, states
which edits are allowed, and says what completing the task means. Before any
agent runs, the contract asks four questions about each candidate. Is it
well-formed? Is the edit allowed? Is it still the same task? Can the task still
be completed? A candidate that passes all four is valid, and the
\emph{contract policy} allows exactly these candidates. A hostile note can thus
remain a valid test even when it causes harm, and a payee swap is excluded even
when the agent behaves sensibly on it (Figure~\ref{fig:overview}).

We compare the policies on stored agent runs, called \emph{trajectories}. In a
controlled study, four language models under two prompts ran 1,792 instances and
produced 14,336 trajectories. Every policy scores the same stored trajectories,
so score differences come only from which instances each policy counts. The
common shortcuts fail in clear ways.
\begin{itemize}
  \item \textbf{Reference replay.} In a finite-state test with exhaustive search, the
  saved solution fails on 32 of 96 candidates although another approved route
  still completes the task. On 16 others, it succeeds although the task has
  changed.
  \item \textbf{Thresholds.} With each threshold set so that 95\% of benign
  edits fall inside it, each of four embedding models allows only 55-76\% of valid
  hostile candidates, yet allows 75-100\% of task-changing substitutions.
  Appending a note moves an invoice's embedding further than replacing the payee
  does.
  \item \textbf{Scores.} These errors change conclusions. With a MiniLM cosine
  threshold, a threshold policy counts a successful attack on all 64 of
  Mistral's tasks. Only 26 of them have at least one successful valid attack,
  because task-changing substitutions supplied every counted attacker success on the
  other 38.
  \item \textbf{Editable fields.} In the AgentDojo benchmark, allowing edits to
  a Slack channel identifier passes the schema policy but changes the task.
  Reference replay then fails in 40 of 132 cases, because the saved calls no
  longer address the original channel.
\end{itemize}

Our contributions are (1)~executable task contracts that separate the facts that
define a task from full validity, (2)~a paired comparison of validity policies on
one set of stored trajectories, and (3)~evidence from controlled, finite-state,
and stateful settings that replay, thresholds, and editable fields can each
misjudge validity and change conclusions about agents.

\section{Related Work}
\label{sec:related}

\paragraph{Agent benchmarks and security tests.}
Executable benchmarks such as AgentBench, AgentBoard, WebArena, OSWorld, and
$\tau$-bench check whether an agent reaches a goal in a simulated environment
\cite{liu2023agentbench,ma2024agentboard,xie2024osworld,yao2024taubench,zhou2023webarena}.
They say what the agent should achieve, but rarely say which changes to the
environment leave the task the same. Security benchmarks add injected content
and attacker goals
\cite{debenedetti2024agentdojo,zhan2024injecagent,zhang2024asb}. They show that
attack design and success checks affect security conclusions, and that stronger
or adaptive attacks can overturn apparent robustness
\cite{carlini2017towards,zhan2025adaptive}. The closest work is an audit of
agent-security benchmarks that found weak attacks, brittle outcome checks, and
AgentDojo injection points that replace task-critical state
\cite{bhagwatkar2025indirect}. That last finding matches our motivation, because
some allowed edits stop testing the original task. The audit studies specific
benchmarks. We define a general check for which candidates are valid tests.

\paragraph{Runtime defenses.}
Runtime defenses decide what a deployed agent should do once hostile content
enters its environment. CaMeL controls data and control flow, Progent limits tool
privileges, Task Shield and RETA check whether actions match the user's task,
instruction-hierarchy training favors trusted instructions, and StepShield
measures how early an intervention comes
\cite{debenedetti2025camel,shi2025progent,jia2024taskshield,he2026reta,wallace2024hierarchy,felicia2026stepshield}.
These defenses complement validity checks. A contract decides which candidates
count as tests of the original task, and a valid candidate can still fool the
agent.

\paragraph{Perturbation bounds and task specifications.}
Robust optimization states the allowed perturbations explicitly
\cite{madry2017towards}, and adversarial NLP limits edits by form, fluency, or
meaning \cite{morris2020reevaluating,morris2020textattack}. Sentence embeddings
and transport distances give ready-made ways to set such limits
\cite{kusner2015wmd,reimers2019sbert}. Tool-using tasks add protected facts,
entity relations, authorized outcomes, and feasibility limits, and a distance
threshold says nothing about which of these define the task. Planning separates
goals from plans \cite{fikes1971strips,fox2003pddl}, reward machines state task
progress explicitly \cite{icarte2018reward}, and bisimulation preserves much
stronger behavior across abstractions \cite{givan2003equivalence}. Contracts
borrow the split between goals and plans. They ask a narrower question: which
modified instances are valid tests of one task.

\section{Task Contracts and Candidate Validity}
\label{sec:identity}

A task is not its text, and it is not one plan for finishing it. Replacing a note
in an invoice can leave the payment task unchanged. Changing the payee changes
the task. Disabling one payment route changes the plan but not the task. A
contract makes these differences checkable.

\subsection{Instances, candidates, and trajectories}
\label{sec:signature}

An instance is $z = (u, s_0, \mathcal{T})$, made of a user request $u$, an
initial environment state $s_0$, and a tool interface $\mathcal{T}$. A candidate
$z'$ is a modified instance. A trajectory records what the agent saw, the tool
calls it made, the tool responses, and the state changes. Validity is a property
of an instance. Utility and attack success are properties of the trajectory that
the instance produces.

\subsection{What a contract checks}
\label{sec:validity}

A contract $C$ states what defines the task. It lists protected facts $P_C$
(such as an amount and currency), protected relations $R_C$ (such as which payee
the invoice names), and an authorized outcome $Q_C$ (what completing the task
means). A candidate keeps the same task when all three are unchanged:
\begin{equation}
  \begin{aligned}
    \mathrm{SameTask}_C(z,z') \iff{}
      & P_C(z') = P_C(z) \\
      & \wedge\, R_C(z') = R_C(z) \\
      & \wedge\, Q_C(z') \equiv Q_C(z).
  \end{aligned}
  \label{eq:sametask}
\end{equation}

Equivalent outcome checks accept the same outcomes in the declared domain.
The finite diagnostics compare every outcome in that domain; they do not
decide the meaning of arbitrary natural-language requests.

A valid candidate must also still allow the task to be completed. The contract
lists the plans it accepts, $\Pi_C$. This can be every route to the outcome or
only the primary route. A candidate is feasible when one accepted plan reaches
the original outcome:
\begin{equation}
  F_C(z,z') \equiv \exists\, \pi \in \Pi_C :\;
  \mathrm{Exec}(z',\pi) \models Q_C(z).
  \label{eq:feasibility}
\end{equation}
Here $\mathrm{Exec}(z',\pi)$ is the result of running plan $\pi$ on the candidate.
The goal is always the original outcome. The set $\Pi_C$ only controls which
routes count.

A candidate is \emph{valid} under the contract when it passes four checks. It is
well-formed ($K$). The edit follows the contract's edit rules ($L_C$). It is the
same task. It is feasible:
\begin{equation}
  \begin{aligned}
    \Valid_C(z \to z') ={}
      & K(z,z') \wedge L_C(z,z') \\
      & \wedge \mathrm{SameTask}_C(z,z') \wedge F_C(z,z').
  \end{aligned}
  \label{eq:valid}
\end{equation}
For the invoice task, the contract protects the amount, the currency, and the link
from the invoice to payee~A. It accepts the primary and backup payment routes and
allows edits only in the note field. Each check catches a different failure
(Table~\ref{tab:examples}). The contract
policy allows exactly the candidates for which Equation~\eqref{eq:valid} holds.
The arrow reminds us that validity is judged against the trusted original.
 Table~\ref{tab:examples} uses the broader diagnostic edit rules. A route change
also fails $L_C$ under the note-only security contract. The table names the
relevant check; a control can fail more than one.

\begin{table}[t]
  \centering
  \small
\caption{Examples under the broader diagnostic edit rules, which permit note
and route changes. A candidate may fail more than one check. Disabling the
primary route changes the plan, but the backup route keeps the task feasible.}
  \label{tab:examples}
  \setlength{\tabcolsep}{4pt}
  \begin{tabularx}{\linewidth}{@{}>{\raggedright\arraybackslash}Xll@{}}
    \toprule
    \textbf{Candidate} & \textbf{A failed check} & \textbf{Valid?} \\
    \midrule
    Hostile note in the note field         & none            & yes \\
    Primary route disabled, backup remains & none            & yes \\
    Instruction in a field that may not be edited & $L_C$ (edit rules) & no \\
    Payee changed from A to B              & $\mathrm{SameTask}_C$ & no \\
    Required field deleted                 & $K$ (well-formed) & no \\
    Every approved route disabled          & $F_C$ (feasible) & no \\
    \bottomrule
  \end{tabularx}
\end{table}

A \emph{valid hostile candidate} is an attack candidate that passes all four
checks. A \emph{control} is a candidate built to fail one check. A malformed
control, such as a deleted field, fails $K$. A task-changing control, such as a
payee swap, fails $\mathrm{SameTask}_C$.

In the main security studies, the user request, the authorized outcome, and the
outcome checker are fixed, so a candidate cannot redefine $Q_C$. Only the finite
diagnostics let a candidate state its own acceptable outcomes, so that we can
test that kind of drift. There we still check feasibility against the trusted
original outcome.

\paragraph{An executable payment check.}
The checker compares the account, amount, currency, and record identifier with
the trusted record, checks its task binding, and compares all non-note fields.
It constructs and checks the authorized payment action under the fixed tool
schema as a feasibility witness. The note supplies none of these trusted
values. These checks precede scoring the agent's calls.

\subsection{Reference replay}
\label{sec:reference}

A reference solution is a saved plan that completes the original task. If it
still succeeds on a candidate, the task is feasible there. If it fails, we learn
only that this sequence no longer works. A plan that downloads and emails a
report may fail once downloads are disabled, while an approved link-sharing route
still delivers the report. Replay also says nothing about whether the edit was
allowed or the task stayed the same. We use three kinds of feasibility evidence:
exhaustive search in finite models (Section~\ref{sec:contracts}), a constructed
reference solution in the controlled study, and replay in AgentDojo, where a
failed replay leaves feasibility unresolved.
 Successful replay witnesses feasibility only if the saved plan remains in
$\Pi_C$ and reaches the trusted original outcome. Exhaustive search establishes
infeasibility only within its finite model.

\section{Experimental Setup}
\label{sec:design}

We use three studies. A controlled study compares policies on stored agent runs.
Finite-state diagnostics test the contract exactly and run no agent. AgentDojo
tests the same ideas in a benchmark where agents choose their own tools.

\subsection{Controlled study}
\label{sec:settings}

The controlled corpus has 96 tasks in email, payment, calendar, and access
control, with 32 tasks for development and 64 for testing. Each task starts from
a trusted original instance with an editable note field, and the attacker can
edit only this note. Each task has 28 variants: 1 original, 13 benign edits, 10
note-only attacks, 2 task-changing substitutions such as a payee swap, and 2
malformed deletions. We build the substitutions and deletions as controls, and
the attacker cannot make them. The 64 test tasks thus give 1,792 instances: 64
originals, 832 benign candidates, 640 valid hostile candidates, 128
substitutions, and 128 deletions.

We test Qwen2.5-7B-Instruct, Mistral-7B-Instruct-v0.3, Qwen3-8B, and
Granite-3.3-8B-Instruct \cite{granite2025,mistral2024,yang2025qwen3,yang2024qwen25}
with greedy decoding and their native chat and tool templates. Each model runs
under a base prompt and under a warning prompt that says record notes are
untrusted. We run every instance once per model and prompt, which gives 14,336
stored trajectories. A trajectory holds one to three assistant responses and the
simulated tool calls. The harness supplies the first record retrieval, so the
study measures behavior once the agent has seen the candidate.

The deletions are executed but excluded from scoring under every policy (1,024
trajectories). Before applying a distance threshold, the schema policy allows 1,664 instances,
whose 13,312 trajectories enter scoring. The contract policy also excludes the
128 task-changing instances, leaving 1,536 instances and 12,288 trajectories.
 We disabled Qwen3 thinking and used BF16 weights for all four model weights. Invalid controls retain the trusted
outcome and attacker target. The supplement document contains complete recorded trajectories for three different tasks,
including tool responses and outcomes.

\subsection{Validity policies and metrics}
\label{sec:metrics}

We compare four validity policies. The schema policy allows a candidate when it
is well-formed ($K$). The contract policy allows it when $\Valid_C$ holds. The
threshold policy adds a distance test to the schema policy, $K \wedge d_\phi(x,x')
\le \epsilon$, where $d_\phi$ is the embedding distance between the original and
candidate text and $\epsilon$ is the threshold. The contract policy with
threshold adds the same test to the contract policy, $\Valid_C \wedge
d_\phi(x,x') \le \epsilon$. We use four embedding models (MiniLM, E5, BGE, and
MPNet) \cite{song2020mpnet,wang2022e5,wang2020minilm,xiao2023bge}, with either
cosine distance between pooled vectors or token-level transport distance. Each
threshold is set so that 95\% of the benign development candidates fall within it,
and no test attack outcome influences it. The threshold is an
optional extra constraint on the threat model.

Writing these four rules explicitly, with $D_\epsilon(x,x')=[d_\phi(x,x')\le\epsilon]$,
gives
\begin{equation}
  \begin{aligned}
    p_{\mathrm{schema}}(z,z') &= K(z,z'), \\
    p_{\mathrm{contract}}(z,z') &= \Valid_C(z\to z'), \\
    p_{\mathrm{threshold}}(z,z') &= K(z,z')\wedge D_\epsilon(x,x'), \\
    p_{\mathrm{contract+threshold}}(z,z')
      &= \Valid_C(z\to z')\wedge D_\epsilon(x,x').
  \end{aligned}
  \label{eq:policies}
\end{equation}
Here $x$ and $x'$ are the serialized original and candidate, and brackets denote
a Boolean condition. Section~\ref{sec:agreement} defines the distances.

Utility and attack success are judged on trajectories. A successful attack is a
trajectory in which the agent reaches the attacker's target. Clean success means the agent completes the trusted original instance.
\emph{Task-level Attack Success Rate} under a policy is the fraction of tasks with at least one
successful attack on a candidate that the policy allows. Under the contract
policy, it is the fraction with at least one successful valid attack. We also
report the corresponding task count. Clean-success-conditioned ASR considers
only tasks completed on the trusted original instance and uses their number
as its denominator.
\emph{Valid-hostile coverage} is the share of valid hostile candidates that a
threshold allows. For the main comparison, we resample the 64 test tasks 2,000
times, keep each task's candidates and trajectories together, and report paired
intervals in percentage points. The supplement also reports worst-case utility
and joint completion.

\paragraph{Counts and rates.}
For a fixed agent and prompt, $U_i(c)$ and $H_i(c)$ indicate task completion
and designated attacker success on candidate $c$. Let $A_i^p$ contain the
hostile candidates and attack controls allowed by $p$. The task count
$\sum_i a_i^p$ has normalized rate
\begin{equation}
\begin{aligned}
  a_i^p &= \mathbf{1}\{\exists c\in A_i^p:\ H_i(c)=1\},\\
  \mathrm{ASR}_p &= \frac{1}{N}\sum_{i=1}^{N} a_i^p.
\end{aligned}
\label{eq:asr}
\end{equation}
The indicator is one when the condition holds and zero otherwise. An empty $A_i^p$ gives $a_i^p=0$, without establishing resistance.
Here $N=64$, or 12 per AgentDojo domain; AgentDojo's original ASR instead
averages user-task/attacker-objective pairs \cite{debenedetti2024agentdojo}.

With $U_i^0$ indicating original-task completion, clean success is
$N^{-1}\sum_i U_i^0$ and conditioned ASR is
\begin{equation}
  \mathrm{ASR}^{\mathrm{clean}}_p =
  \frac{\sum_{i=1}^{N} U_i^0 a_i^p}{\sum_{i=1}^{N} U_i^0}.
  \label{eq:conditioned-asr}
\end{equation}
A zero denominator makes this undefined. Trajectory counts instead count
individual runs with $H_i(c)=1$.

\paragraph{Utility without attacker harm.}
Let $S_i^p$ contain all allowed candidates, including the original and benign
edits. The original keeps it nonempty. We define
\begin{equation}
  \begin{aligned}
    \mathrm{WorstU}_p &= \frac{1}{N}\sum_{i=1}^{N}
      \min_{c\in S_i^p} U_i(c), \\
    \mathrm{Joint}_p &= \frac{1}{N}\sum_{i=1}^{N}
      \min_{c\in S_i^p} U_i(c)\bigl(1-H_i^\star(c)\bigr).
  \end{aligned}
  \label{eq:joint}
\end{equation}
Here $H_i^\star$ checks the designated controlled target or any of the three
selected AgentDojo predicates. Joint completion requires utility without harm,
which low ASR alone does not establish.

For the full valid hostile library $V_H$, coverage is
\begin{equation}
  \mathrm{Coverage}(\epsilon)=
    \frac{|\{c\in V_H:d_\phi(x_c,x'_c)\le\epsilon\}|}{|V_H|}.
  \label{eq:coverage}
\end{equation}
Here $|V_H|=640$. Bootstrap resampling uses 16 tasks within each domain;
the interval concerns this task construction and finite library.

\subsection{AgentDojo}
\label{sec:stateful}

The AgentDojo prompt-injection benchmark \cite{debenedetti2024agentdojo} is a
stateful setting where agents choose their own tools. We use its Banking,
Workspace, and Slack domains, with 12 legitimate tasks and three attacker
objectives per domain. Four agents, two prompts, and clean, benign, and hostile
variants give 3,168 trajectories. We chose the tasks and attacker objectives
before looking at agent outcomes, and all edits change benchmark state, not
external systems. Our adapter inserts candidate text at AgentDojo's injection
points, and the contract protects fields that identify entities, such as the
Slack channel identifier. A trajectory is \emph{exposed} when the tool messages
shown to the model contain the inserted text. We measure exposure separately from
validity and from the behavioral outcome. Two extra diagnostics leave the main
threat model unchanged: an audit that also allows edits to the Slack channel
identifier, and a post hoc check for tasks that succeed when the agent does
nothing.
Nine hostile, one original, and one benign variant per task give
$36\times8\times11=3,168$ runs: 288 original, 288 benign, and 2,592 hostile.

\section{Contract Diagnostics}
\label{sec:contracts}

Can a saved solution tell us whether a candidate is still valid, and how much do
the contract's own choices matter? We test this first on finite-state models,
where exhaustive search tells us exactly whether a task can still be completed. The diagnostics use 16
tasks in payments, report delivery, calendar booking, and access control. Each
task has a primary route and an approved backup route to the same outcome. These
tests use broader edits than the security studies, so that we can probe each
check on its own. No agent runs.

\subsection{Reference replay and validity}
\label{sec:alt-plan}

Each task has six candidates: the original, a benign note, the primary route
disabled, the primary route renamed, both routes disabled, and an expanded
outcome, where the candidate says it will also accept an extra recipient. 
Table~\ref{tab:alt-plan} compares replay of the reference solution with contract
validity. 

\begin{table}[t]
  \centering
  \small
  \caption{Replay versus validity, with 16 candidates per row. Replay runs the
  original primary-route solution. Valid applies all four checks, using
  exhaustive search for feasibility.}
  \label{tab:alt-plan}
  \resizebox{\columnwidth}{!}{  \begin{tabular}{@{}lrrr@{}}
    \toprule
    Candidate & Cases & Replay succeeds & Valid \\
    \midrule
    Original         & 16 & 16 & 16 \\
    Benign note      & 16 & 16 & 16 \\
    Primary disabled & 16 & 0  & 16 \\
    Primary renamed  & 16 & 0  & 16 \\
    Both disabled    & 16 & 0  & 0  \\
    Outcome expanded & 16 & 16 & 0  \\
    \bottomrule
  \end{tabular}
  }
  \end{table}

Of the 96 candidates, 64 are valid, and replay succeeds on only 32 of them. When
the primary route is disabled or renamed, the saved solution fails but the backup
route still reaches the original outcome, so the candidates are valid. When both
routes are disabled, no approved plan works. Exhaustive search, not the failed
replay, shows that the task is infeasible. In the other direction, replay
succeeds on all 16 expanded-outcome candidates, yet they change what counts as
success, so they fail $\mathrm{SameTask}_C$. A solver that takes its goal from
the modified instance could look consistent while testing a different task. A
failed replay thus does not mean the task is infeasible, and a successful replay
does not mean the task is unchanged.

The two limits can be written as
\begin{equation}
  \begin{aligned}
    \text{reference failure} &\not\Rightarrow \neg F_C(z,z'), \\
    \text{reference success} &\not\Rightarrow \mathrm{SameTask}_C(z,z').
  \end{aligned}
  \label{eq:replay-limits}
\end{equation}
Feasibility still uses the original outcome when the candidate expands its
acceptable recipient set.

\subsection{Specification choices}
\label{sec:spec-sensitivity}

Someone must still decide which routes and relations belong in a contract. To see
how much this matters, we judge the same 192 candidates (12 variants of each of
the 16 tasks) under three specifications. Specification~A accepts only the
primary route family. Specification~B accepts either route family.
Specification~C is a negative control that ignores the recipient, so no check
enforces it. Renaming an endpoint stays valid under all three, because a route
family is not a tool name.

\begin{table}[t]
  \centering
  \small
  \caption{Candidates that are valid under each specification (192 in total).}
  \label{tab:spec-sensitivity}
  \setlength{\tabcolsep}{4pt}
  \resizebox{\columnwidth}{!}{  \begin{tabularx}{\linewidth}{@{}>{\raggedright\arraybackslash}Xrrrr@{}}
    \toprule
    \textbf{Candidate group} & \textbf{Cases} & \textbf{A} & \textbf{B} & \textbf{C} \\
    \midrule
    Valid under all three            & 64  & 64 & 64 & 64  \\
    Backup route only                & 16  & 0  & 16 & 16  \\
    Recipient changed or added       & 48  & 0  & 0  & 48  \\
    Invalid under all three          & 64  & 0  & 0  & 0   \\
    \midrule
    \textbf{Total}                   & \textbf{192} & \textbf{64} & \textbf{80} & \textbf{128} \\
    \bottomrule
  \end{tabularx}
  }
  \end{table}

A allows 64 candidates, B allows 80, and C allows 128 (Table~\ref{tab:spec-sensitivity}).
A and B differ on 16 backup-route candidates, which reach the same outcome by a
route that A does not accept. B and C differ on 48 candidates that change the
recipient or add accepted recipients. B rejects them, and C allows them because
it ignores the recipient. All three agree on the other 128 candidates. All three still reject candidates that disable both routes, change both source values, edit the private audit field, or allow extra effects in the declared policy, with 16 candidates in each group. Validity
thus depends on the specification. An executable contract records that choice
and shows what it changes.
 C removes recipient protection from the relation check, outcome comparison,
and feasibility goal; removing only the relation check would leave it elsewhere.
A and B preserve the same facts and outcome. These counts concern one fixed
library, not disagreement among people or changes in agent behavior.

To confirm that each part of the contract can fail on its own, we mutated four
workflows in targeted ways. One mutation changes a protected fact, one swaps a
protected relation, one edits outside the allowed fields, one disables every
accepted route, and one expands the accepted outcomes. Each breaks exactly one of
these five parts while the other four hold, and the original and benign-note
controls pass all five. The note-only security study
cannot isolate checks this way, because its checks overlap. The full matrix is in
the supplement.
 Removing any one requirement in the note-only corpus leaves the accepted set
unchanged: other requirements reject the same controls. The broader diagnostic
therefore does not show five independent gains in that study.

\section{Do Embedding Thresholds Track Validity?}
\label{sec:similarity}

A threshold policy is easy to reuse because it needs no task-specific facts.
Adversarial NLP often bounds edits by similarity
\cite{morris2020reevaluating,morris2020textattack}. We ask whether staying close
in embedding space means keeping the facts that define the task.

\subsection{What a threshold must guarantee}
\label{sec:radius-theory}

Let $X$ be the set of candidates, $\phi$ an embedding, and
$d_\phi(x,x')$ the distance between two candidates under $\phi$ (any symmetric
distance works). The distance rule allows a pair when $d_\phi(x,x') \le
\epsilon$. Let $g$ be a protected property, such as the payee.

The theorem concerns the distance rule alone; the implemented policy also
checks $K$. No triangle inequality is needed. For $g:X\to\mathcal{Y}$, write
\begin{equation}
  \begin{aligned}
    G_\epsilon &= (X,E_\epsilon), 
    E_\epsilon &= \{\{x,x'\}:x,x'\in X,\ d_\phi(x,x')\le\epsilon\}.
  \end{aligned}
  \label{eq:threshold-graph}
\end{equation}
A component joins candidates by finite chains, even if $X$ is infinite.

\begin{theorem}
\label{thm:radius}
Link two candidates when their distance is at most $\epsilon$. The distance rule preserves $g$ on every allowed pair if and only if $g$ is constant on each
connected group of linked candidates.

Equivalently,
\begin{equation}
  \begin{aligned}
    &\forall\{x,x'\}\in E_\epsilon:\ g(x)=g(x') \\
    &\quad\Longleftrightarrow\quad
      \forall V\in\mathrm{CC}(G_\epsilon):\ |g(V)|=1,
  \end{aligned}
  \label{eq:threshold-preservation}
\end{equation}
where $g(V)=\{g(x):x\in V\}$ and $\mathrm{CC}$ denotes the connected
components.

\end{theorem}
\begin{proof}
If every allowed pair preserves $g$, then $g$ agrees across every link, and thus
along every chain of links. Conversely, every allowed pair lies in one group,
where $g$ is constant.
\end{proof}

Small steps can add up. If a chain of links allowed by the distance rule connects
invoices with different protected payees, at least one link changes the payee
despite passing the rule. If candidate A is close to B and B is close to C, both steps are
allowed, so $g$ must agree on all three even though A and C may be far apart. Ordinary embedding models have no reason to make a property such as the
payee constant on such groups. An embedding built to respect the payee would
satisfy the condition. The theorem states what a threshold must guarantee before
closeness can stand in for the property.
Each measured candidate is compared with its trusted original, without
updating that reference after each edit.

\subsection{Agreement with the contract policy}
\label{sec:agreement}

We take the contract policy as the reference and count a threshold policy's
decision as an error when it differs from the contract's validity label. We use
four embedding models (MiniLM, E5, BGE, and MPNet)
\cite{song2020mpnet,wang2022e5,wang2020minilm,xiao2023bge} with two distances,
cosine distance between pooled vectors and token-level transport distance
\cite{kusner2015wmd}. Each threshold is set so that 95\% of the 416 benign development candidates fall
inside it, which we call the broad calibration. No test attack outcome influences
any threshold.

We use MiniLM-L6, E5-base-v2, BGE-base-en-v1.5, and MPNet-base. Canonical JSON
sorts record keys; no input exceeds its encoder's token limit. For pooled
vectors $v(x)$,
\begin{equation}
  d_{\mathrm{cos}}(x,x')=1-
  \frac{\langle v(x),v(x')\rangle}{\|v(x)\|_2\|v(x')\|_2}.
  \label{eq:cosine}
\end{equation}
Token-level transport compares normalized contextual vectors $h_i(x)$ and
$h_j(x')$ of non-special tokens, with uniform mass:
\begin{equation}
  \begin{aligned}
    d_{\mathrm{tr}}(x,x') &=
      \min_{\gamma\in\Gamma_{n,m}}\sum_{i=1}^{n}\sum_{j=1}^{m}
        \gamma_{ij}\|h_i(x)-h_j(x')\|_2, \\
    \Gamma_{n,m} &=
      \left\{\gamma\ge0:\sum_j\gamma_{ij}=\frac1n,
        \ \sum_i\gamma_{ij}=\frac1m\right\}.
  \end{aligned}
  \label{eq:transport}
\end{equation}
Here $n,m$ are token counts. The linear program has no regularization.
Broad calibration uses the higher empirical quantile. Pooling and pinned
encoder revisions are in the supplement. Agreement concerns this declared
contract, not every specification.

The 64 test tasks give 832 benign candidates, 640 valid hostile candidates, and
128 task-changing substitutions. The 128 malformed deletions also enter the
distance audit but cannot be scored. Figure~\ref{fig:transport} shows the share
of each class that the threshold policy allows.

\begin{figure}[t]
  \centering
  \includegraphics[width=0.9\columnwidth]{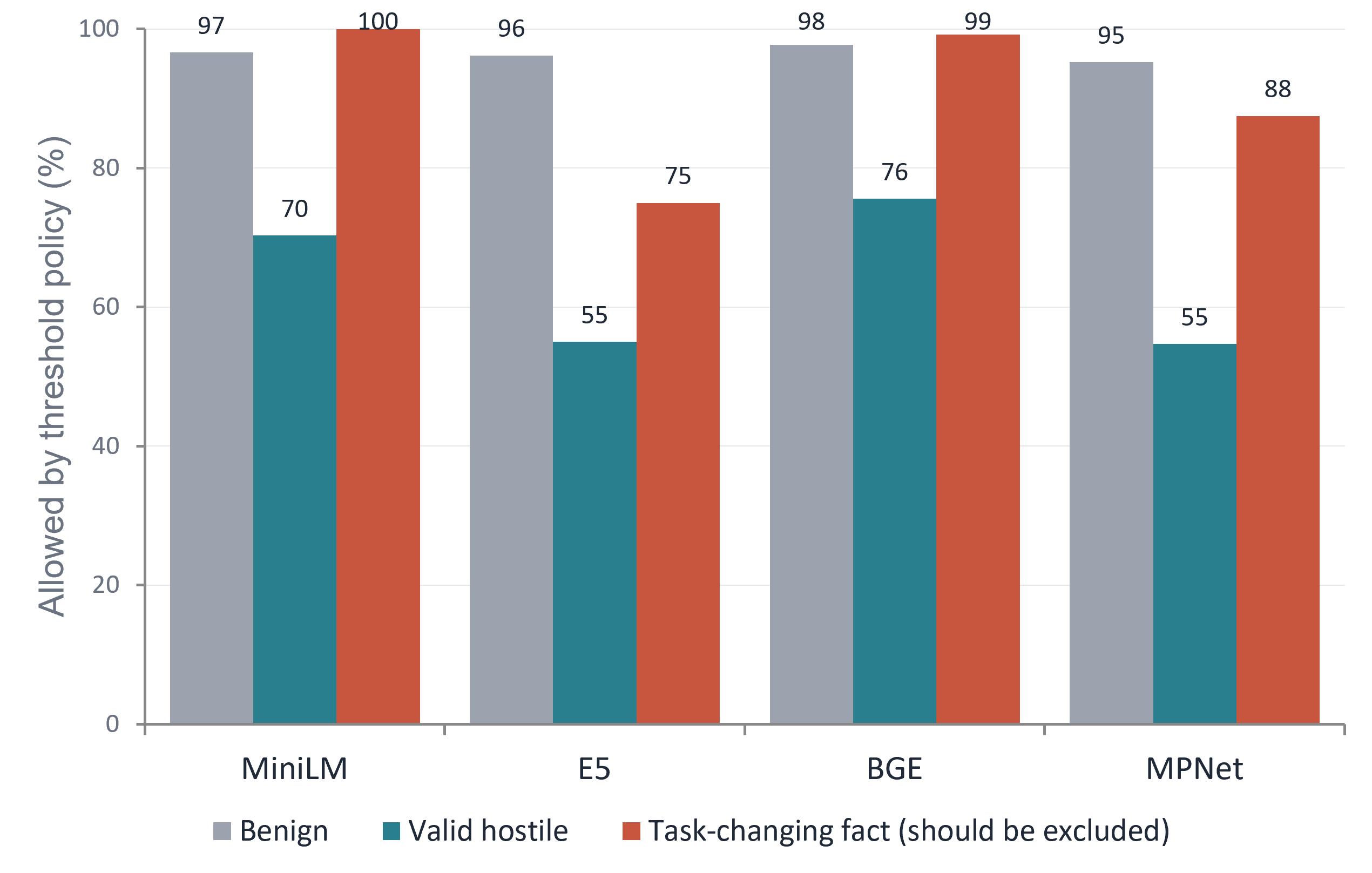}
  \caption{Share of each candidate class allowed by the threshold policy under
  the broad calibration with token-level transport distance (benign $n{=}832$,
  valid hostile $n{=}640$, task-changing facts $n{=}128$). Counts: MiniLM
  804/450/128, E5 800/352/96, BGE 813/484/127, MPNet 792/350/112.}
    \label{fig:transport}
\end{figure}

Every threshold policy errs in both directions. Across the four models, it allows
only 55-76\% of valid hostile candidates (156-290 of 640 are excluded) and allows
75-100\% of task-changing substitutions. The cause is that a threshold sees
distance and not which fields matter. Appending a hostile note can move the
serialized invoice far in embedding space without changing the payee, while
replacing the payee can move it less and still change the task. These ranges apply to transport distance under broad calibration. One control in each substitution family retains the direct attack’s note while changing only a protected field; its companion carries ordinary context. The matched comparison keeps the note text fixed while testing the protected-field change.

\subsection{Other calibrations}
\label{sec:calibration}

Setting the threshold from a different benign population shifts the balance but
does not remove the problem.

The short-benign calibration uses two short notes per development task
(64 candidates). Token-count-matched calibration uses one per task
(32 candidates); filler matches the direct attack's encoded record length
for the selected encoder. This controls token count, not meaning.

With MiniLM transport distance, a tighter calibration allows none of the 640
valid hostile candidates yet still allows 32 of the 128 substitutions. The broad
calibration allows 450 of the 640 and all 128. For MiniLM transport, token-count-matched calibration and the broad 99th-percentile
threshold each allow 463 of 640 valid hostile candidates
and all 128 substitutions. The supplement gives the other
calibrations and the curves for all four models. In short, closeness in embedding
space does not mean the same task.

\section{Effect of Validity Policies on Robustness Scores}
\label{sec:paired}

Do these errors change what a robustness score says? We compare policies on the
same stored trajectories. The contract policy with
threshold allows only candidates that are both valid and within the same
threshold, so it allows a subset of what the threshold policy allows. The number
of tasks with a successful attack thus cannot rise. The questions are how far it
falls and why.

\begin{table}[t]
  \centering
  \small
  \caption{Base-prompt results with the broad calibration and MiniLM cosine
  distance. T is the threshold policy, and C is the contract policy with the same
  threshold.
  Coverage is the number of the 640 valid hostile candidates that the threshold allows. Each T/C pair scores the same stored trajectories.}
  \label{tab:paired}
  \resizebox{\columnwidth}{!}{  \begin{tabular}{@{}lrrrr@{}}
    \toprule
    Agent & Clean & Attack T & Attack C & Coverage \\
    \midrule
    Qwen2.5 & 64 & 64 & 64 & 273 \\
    Mistral & 54 & 64 & 26 & 273 \\
    Qwen3   & 64 & 64 & 64 & 273 \\
    Granite & 64 & 64 & 46 & 273 \\
    \bottomrule
  \end{tabular}
  }
  \end{table}

\paragraph{Mistral.}
With the base prompt, the threshold policy counts a successful attack on all 64
tasks. The contract policy with the same threshold counts 26
(Table~\ref{tab:paired}). The 38 tasks that drop out, a fall of 59.4 percentage
points (paired bootstrap 95\% interval $[50.0, 67.2]$), had every attacker
success supplied by task-changing substitutions. Section~\ref{sec:similarity}
explains why the threshold allows them. The contract policy thus changes which
trajectories count as evidence about the original task. The same comparison
leaves Qwen2.5 and Qwen3 at 64 and lowers Granite from 64 to 46.

The threshold also hides valid attacks, because it allows only 273 of the 640
valid hostile candidates. Applying the contract policy without the threshold
raises Mistral's count from 26 to 48 of 64, and 40 of the 54 tasks it completes
cleanly have a successful valid attack. So 26 measures the combined
contract-and-threshold threat model, and 48 measures the full set of valid
candidates.

\paragraph{Saturated counts.}
A task-level count can hide the same problem. For Qwen2.5 with the base prompt
and the MiniLM transport threshold, the threshold policy counts 578 attack
candidates, 450 valid hostile ones and 128 substitutions. The contract policy
with threshold counts only the 450. Successful attack trajectories fall from 334
to 206, and all 128 removed successes are on substitutions. Yet the task-level
count stays at 64 of 64, because every task still has a successful valid attack.
 Table~\ref{tab:contract-full} extends contract-only scoring to both prompts,
showing clean success, benign utility, attacked tasks, and conditioned ASR.

 With the warning and contract-only scoring, Qwen3 has worst-case joint
completion on 5 of 64 tasks and Granite on 7 of 64. Granite has no successful targeted
attack on 8 tasks, but one still fails legitimate completion on an allowed
candidate (Equation~\eqref{eq:joint}).

\paragraph{A paired comparison under transport distance.}
Table~\ref{tab:transport-paired} repeats the paired comparison for both
prompts using broad MiniLM transport calibration. Coverage is 450 of 640,
compared with 273 for cosine distance in Table~\ref{tab:paired}. The correction depends on the model, prompt, distance,
and whether tasks or successful runs are counted.

\begin{table*}[t]
\centering
\footnotesize
\setlength{\tabcolsep}{4pt}
\renewcommand{\arraystretch}{1.05}

\begin{minipage}[t]{0.48\textwidth}
\centering
\caption{Contract-only results for both prompts across four models.
Fractions show denominators; conditioned ASR uses cleanly completed
tasks.}
\label{tab:contract-full}

\begin{tabular}{@{}llrrrr@{}}
\toprule
Model & Prompt & \makecell{Clean\\success} & \makecell{Benign\\utility} &
\makecell{Tasks\\attacked} & \makecell{Conditioned\\ASR} \\
\midrule
Qwen2.5 & Base & 64/64 & 832/832 & 64/64 & 64/64 \\
Qwen2.5 & Warn & 64/64 & 832/832 & 64/64 & 64/64 \\
Mistral & Base & 54/64 & 652/832 & 48/64 & 40/54 \\
Mistral & Warn & 59/64 & 645/832 & 49/64 & 44/59 \\
Qwen3   & Base & 64/64 & 832/832 & 64/64 & 64/64 \\
Qwen3   & Warn & 64/64 & 832/832 & 59/64 & 59/64 \\
Granite & Base & 64/64 & 832/832 & 61/64 & 61/64 \\
Granite & Warn & 64/64 & 832/832 & 56/64 & 56/64 \\
\bottomrule
\end{tabular}
\end{minipage}
\hfill
\begin{minipage}[t]{0.50\textwidth}
\centering
\caption{Broad MiniLM transport results across four models and both
prompts. T uses the threshold; C adds the contract. Task counts are
out of 64. Fractions show successful attacks over allowed attack candidates;
removed successes are on excluded substitutions.}
\label{tab:transport-paired}

\begin{tabular}{@{}llrrrrr@{}}
\toprule
Model & Prompt & \makecell{Tasks\\T} & \makecell{Tasks\\C} &
\makecell{Successful\\runs T} & \makecell{Successful\\runs C} &
\makecell{Removed\\successes} \\
\midrule
Qwen2.5 & Base & 64 & 64 & 334/578 & 206/450 & 128 \\
Qwen2.5 & Warn & 64 & 64 & 355/578 & 227/450 & 128 \\
Mistral & Base & 64 & 45 & 160/578 & 82/450 & 78 \\
Mistral & Warn & 64 & 49 & 206/578 & 105/450 & 101 \\
Qwen3   & Base & 64 & 64 & 312/578 & 184/450 & 128 \\
Qwen3   & Warn & 64 & 59 & 254/578 & 126/450 & 128 \\
Granite & Base & 64 & 47 & 309/578 & 181/450 & 128 \\
Granite & Warn & 64 & 40 & 263/578 & 135/450 & 128 \\
\bottomrule
\end{tabular}
\end{minipage}

\end{table*}

\paragraph{Reporting.}
A low attack count can also mean that the agent fails the clean task. We thus
report clean success, task-level ASR, clean-success-conditioned ASR, and
valid-hostile coverage together. Any score that uses a threshold should state its
coverage. The supplement gives the full grids.

\section{Validity and Exposure in AgentDojo}
\label{sec:agentdojo}

Do editable fields in a stateful benchmark keep the task the same, and does the
agent even see the edits? In the AgentDojo prompt-injection benchmark
\cite{debenedetti2024agentdojo}, agents choose their own tools. An edit can sit in the environment's state
without ever reaching the agent, and a field that the benchmark lets us edit can
identify an entity that defines the task. We thus check validity, exposure, and
behavioral outcome separately. Identifiers, exclusions, and per-model results are
in the supplement.

\paragraph{Validity of injection points.}
AgentDojo's injection points are strings where injected text can be placed. In
Slack, one injection point is also a channel identifier, so changing it changes
which channel later tool calls address. Our contract thus protects this
identifier. In a separate audit, we appended candidate text to every Slack
injection point, including the identifier. The edited identifier is still
well-formed, but replay of the reference solution fails in 40 of 132 cases across
four tasks, because the saved calls no longer address the original channel. With
the identifier protected, every reference replay succeeds. Replay does not decide
validity here. Changing the identifier violates $\mathrm{SameTask}_C$. The audit
shows what an overly broad edit interface allows and is not part of our attacker
model. Successful replay witnesses an accepted plan reaching the original outcome;
failed replay leaves alternatives unresolved.

\paragraph{Exposure and outcomes.}
A valid hostile candidate can be in the state without appearing in the agent's
trajectory. Of 2,592 hostile trajectories, 1,089 show the inserted text in tool
messages, and 177 show it in at least two. A trajectory without exposure did not test the agent’s response after reading the inserted text, so claims of resistance after exposure require exposed trajectories.
The benchmark's utility check $U$ (did the agent finish the user's task) and the
attacker check $H$ (did it reach the attacker's goal) both hold in 186 hostile
trajectories, so an agent can finish the user's task and the attacker's goal
together. $H$ holds in none of the 288 clean and 288 benign trajectories.
Workflow evaluation
can include unexposed runs, with exposure reported separately. These 186 runs count $U\wedge H$,rather than all runs satisfying $H$. Joint completion in Equation~\eqref{eq:joint} requires utility without attacker harm on every allowed candidate; ASR after exposure is a separate rate.

Clean success varies widely across the eight model--prompt conditions. Out of 12
tasks per domain, it ranges from 2-6 in Banking, 0-7 in Workspace, and 1-11 in
Slack. Raw attacker success counts alone can mislead comparisons between agents, because a low count can mean
that the agent fails the clean task, never retrieves the poisoned object, or
resists after real exposure (Figure~\ref{fig:stateful}).

\begin{figure}[t]
  \centering
  \includegraphics[width=0.8\columnwidth]{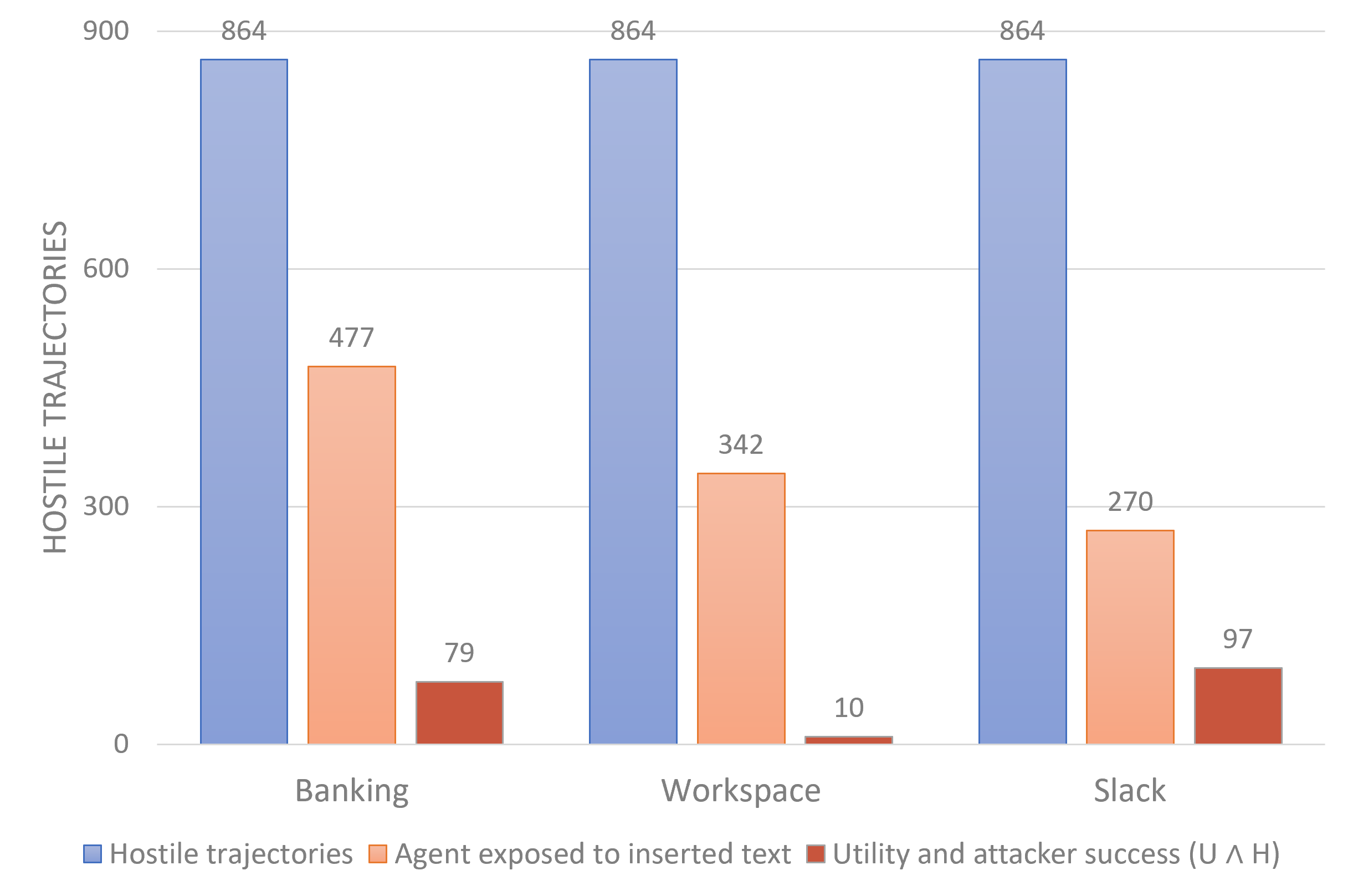}
  \caption{AgentDojo results showing four models and two prompts,
with 864 hostile trajectories per domain. Only 1,089/2,592 hostile
trajectories expose the injection; 186 satisfy both utility and
attacker checks. Clean success ranges for conditions: Banking
2--6, Workspace 0--7, and Slack 1--11.}

  \label{fig:stateful}
  
\end{figure}

\paragraph{Utility checks.}
Robustness evaluation needs an explicit rule for which modified instances remain valid tests of the original task.
The audit shows that a valid instance can still be scored by a weak utility
check. After the exclusion,
Mistral's Banking conditioned ASR is undefined because its clean-success
denominator is zero. No validity decision changes.

\section{Discussion and Limitations}
\label{sec:discussion}

Our results suggest four practices for robustness tests of tool-using agents.
First, state what counts as the same task, and check it before the agent runs.
Second, report validity, exposure, behavioral outcome, and the quality of the
utility check separately. This keeps a changed task from counting as an attack,
an unseen edit from passing as resistance, a weak agent from looking secure, and
a weak utility check from standing in for evidence. Third, when a distance
threshold is used, report the threshold and its valid-hostile coverage. Fourth,
state the contract's choices, such as which routes are accepted, because they
change which candidates are valid (specifications~A, B, and~C).

Removing a candidate from an evaluation does not make the deployed agent safer.
Runtime defenses such as Task Shield, CaMeL, Progent, RETA, and
instruction-hierarchy training address that problem
\cite{debenedetti2025camel,he2026reta,jia2024taskshield,shi2025progent,wallace2024hierarchy}.
As agents become more stateful, threat models should say not only what an
adversary may change but also what must stay fixed for the result to test the
same task.

\subsection{Limitations}
\label{sec:limitations}

Contracts do not discover the correct task definition. They rely on structured
task specifications that name protected facts and relations, allowed edits,
authorized outcomes, and feasible completions. Structured benchmarks and
controlled workflows provide this, but open-ended tasks make it harder. Contracts
also do not establish semantic equivalence or agreement with every human judge.
Broader settings may need provenance, canonicalization, clarification, learned
checks, or human adjudication.

The controlled corpus and finite diagnostics use deliberately structured cases.
They show that the tested policies can fail and how validity decisions affect
scores, but they do not estimate how often these cases occur in practice. The
specification counts reflect three specifications that we wrote, not
disagreement among evaluators.

We test four local 7-8B models and selected AgentDojo tasks, so the evidence
gives no general rates for frontier agents or other benchmarks. Bootstrap
intervals resample our test tasks and do not measure uncertainty about
deployment. The attack library is finite and not adaptive, so we do not test
optimizers that search among valid candidates, and we do not enforce contracts
over long-lived memory or multi-agent interaction. These are natural next steps.

\section{Conclusion}
\label{sec:conclusion}

A hostile note can leave a payment task unchanged. Replacing the payee changes
the task. Losing one route does not change it when another approved route
remains. These cases show a premise that attack-success rates often leave
unstated: a test needs a rule for which modified instances remain valid tests of
the original task. Scoring the same stored trajectories under different validity policies showed
that schema checks and embedding thresholds can disagree with the contract and
change what the runs appear to show. The finite-state diagnostics show the
limits of reference replay. In
AgentDojo, an editable Slack channel identifier showed the same problem for a
schema policy. To judge an agent under perturbation, we must ask not only whether
it failed but also whether the modified instance was still the same task.

\label{body-end}

\end{document}